%% file: nt-cce.tex
\documentclass[11pt]{article}

\usepackage{amsmath,amstext,amsthm,amssymb,amsfonts}

\usepackage{fullpage}
\usepackage{url}
\usepackage{bbm}
\usepackage{graphicx}
\usepackage{psfrag,subfig}
\usepackage{verbatim}
\usepackage{algorithmic}
\usepackage{algorithm}

\newtheorem{theorem}{Theorem}

\newtheorem{corollary}{Corollary}
\newtheorem{definition}{Definition} 

\DeclareMathOperator*{\argmax}{arg\,max}
\DeclareMathOperator*{\argmin}{arg\,min}

\newcommand{\E}{\mathbb{E}}

\newcommand{\OPT}{\textrm{OPT}}
\newcommand{\OP}{$\mathrm{NT}$ }
\newcommand{\AOP}{$\mathrm{ANT}$}
\newcommand{\M}{\mathcal{M}}
\newcommand{\A}{\mathcal{A}}

\begin{document}
\title{Finding Any Nontrivial Coarse Correlated Equilibrium Is Hard}
\author{Siddharth Barman\thanks{California Institute of Technology. E-mail: {\tt barman@caltech.edu.}}\and Katrina Ligett\thanks{California Institute of Technology. E-mail: \tt{katrina@caltech.edu.}}}
\date{}
\maketitle

\begin{abstract}

One of the most appealing aspects of the (coarse) correlated equilibrium concept is that natural dynamics quickly arrive at approximations of such equilibria, even in games with many players.
In addition, there exist polynomial-time algorithms that compute exact (coarse) correlated equilibria. In light of these results, a natural question is how {\em good} are the (coarse) correlated equilibria that can arise from any efficient algorithm or dynamics.

In this paper we address this question, and establish strong negative results. In particular, we show that in multiplayer games that have a succinct representation, it is \rm{NP}-hard to compute any coarse correlated equilibrium (or approximate coarse correlated equilibrium) with welfare strictly better than the \emph{worst} possible. The focus on succinct games ensures that the underlying complexity question is interesting; many multiplayer games of 
interest are in fact succinct. Our results imply that, while one can efficiently compute a coarse correlated equilibrium, one cannot provide any nontrivial welfare guarantee for the resulting equilibrium, unless $\rm{P=NP}$. We show that analogous hardness results hold for correlated equilibria, and persist under the egalitarian objective or Pareto optimality. 

To complement the hardness results, we develop an algorithmic framework that identifies settings in which we can efficiently compute an approximate correlated equilibrium with near-optimal welfare. We use this framework to develop an efficient algorithm for computing an approximate correlated equilibrium with near-optimal welfare in  aggregative games. 

\end{abstract}

\maketitle

\input{intro}

\input{notation}

\input{hardness}

\input{algo}

\section*{Acknowledgements}
The authors thank Nikhil Bansal for helpful discussions. This work was supported by NSF grants CNS-0846025, CCF-1101470, and CNS-1254169, along with a Microsoft research faculty fellowship, a Google faculty research award, and a Linde/SISL postdoctoral fellowship. Katrina Ligett gratefully acknowledges the support of the Charles Lee Powell Foundation.

\bibliographystyle{plain}
\bibliography{apprx-opt-CE}

\end{document}

%% file: intro.tex

\section{Introduction}


Equilibria are central solution concepts in game theory, and questions related to the complexity of equilibrium computation have formed a major thread of research in algorithmic game theory. Arguably the most important equilibrium concepts are the Nash equilibrium~\cite{nash1951non}, correlated equilibrium~\cite{aumann1974subjectivity}, and coarse correlated equilibrium~\cite{hannan1957approximation}. These solution concepts denote distributions over players' action profiles at which no player can benefit by unilateral deviation, and hence represent stable choices of distributions over player actions. Specifically, a Nash equilibrium is defined to be a product of independent distributions (one for each player); correlated and coarse correlated equilibria are general (joint) probability distributions (see Section~\ref{sect:notation} for formal definitions).  


While computation of Nash equilibria has in recent years been shown to be computationally hard, even in games with two players~\cite{chen2009settling}, the news for correlated equilibria (CE) and coarse correlated equilibria (CCE) has been more positive. Even in games with many players, there exist a number of natural dynamics that quickly converge to these solution concepts (see, e.g.,~\cite{littlestone1994weighted, foster1998asymptotic, hart2000simple, blum2007external}). In particular, these dynamics induce efficient computation of approximate\footnote{A probability distribution over the players' action profiles is said to be an $\varepsilon$-approximate equilibrium if for any player unilaterally deviating increases utility, in expectation, by at most $\varepsilon$.} CE and CCE in multiplayer games; by contrast, computation of approximate Nash equilibria is computationally hard in multiplayer games~\cite{R}. In fact, {\em exact} CE and CCE are efficiently computable in many classes of multiplayer games~\cite{PR,JLB}.
 


Another significant thread of research in algorithmic game theory has been the study of the {\em quality} of equilibria, often as measured by the social welfare of the equilibrium or its ratio to the social welfare of the socially optimal outcome (c.f. the extensive literature on the price of anarchy (PoA) \cite{AGTbook}). Given that we know it is possible to efficiently compute CE and CCE, it is natural to ask {\em how good} are the equilibria we can efficiently compute? For example, do existing efficient dynamics find the best such equilibria, or at least ones that approximately optimize the social welfare?
Since the gap between the worst and the best equilibria (CE or CCE), in terms of social welfare, can be large in natural games
(see, e.g.,~\cite{lee2013improved, bilo2013price}), 
it is interesting to understand if there exist efficient dynamics or algorithms that avoid---at least to some extent---the bad outcomes.
More generally, one can pose the question of efficiently finding CE and CCE \emph{that optimize an objective} (such as the sum of players' utilities, i.e., the social welfare). 

 
In their notable work, Papadimitriou and Roughgarden~\cite{PR} show that determining a \emph{socially optimal} CE is \rm{NP}-hard, in a number of succinct multiplayer games. This result intuitively follows from the fact that determining an action profile with maximum welfare---i.e., solving the problem of welfare optimization even without equilibrium constraints---is \rm{NP}-hard in general. The hardness result of~\cite{PR}  leaves open the question of computing \emph{near-optimal} CE/CCE, i.e., whether there exist efficient algorithms that compute CE/CCE with welfare at least, say, $\alpha$ times the optimal, for a nontrivial approximation ratio $\alpha \leq 1$. This question forms the basis of the present work.

\emph{Technical Aside (succinct games):} We note that in general multiplayer games the size of the normal form representation, $N$, is exponentially large in the number of players; one can compute a CE/CCE that optimizes a linear objective by solving a linear program of size polynomial in $N$, and hence the computational complexity of equilibrium computation is not interesting for general games. However, most games of interest---such as graphical games, polymatrix games, congestion games, local effect games, network design games, anonymous games, and scheduling games---admit a succinct representation (wherein the above-mentioned linear program can be exponentially large in the size of the representation), and hence it is such succinctly representable games that we (and previous works) study.\footnote{Note that the optimization problem does not become simpler if, instead of a succinct game, one is given access to a game via a black box which, when given an action profile $a$ as a query, returns the utilities of all the players at $a$.}





\paragraph{Results} In this paper we establish that, unless $\rm{P=NP}$, there does not exist any efficient algorithm that computes a CCE with welfare better than the \emph{worst possible} CCE,  in succinct multiplayer games (Theorem~\ref{thm:op}). 
We also establish similar hardness results for computing equilibria under the egalitarian objective or Pareto-optimality. 

Analogous hardness results hold for CE. We note that a classical interpretation of a CE is in terms of a mediator who has access to the players' payoff functions and who draws outcomes from a correlated equilibrium's joint distribution over player actions and privately recommends the corresponding actions to each player. The equilibrium conditions ensure that no player can benefit in expectation by unilaterally deviating from the recommended actions. Therefore, the problem we study here is exactly the computational complexity of the problem that a mediator faces if she wishes to maximize social welfare.


We also extend the hardness result to approximate CE and CCE  (Theorem~\ref{thm:aop}). Therefore, while one can efficiently compute an approximate CE/CCE in succinct multiplayer games, one cannot provide any nontrivial welfare guarantees for the resulting equilibrium (unless $\rm{P=NP}$). 


In addition, we show that this hardness result also holds specifically for potential games (generally considered to be a very tractable class of games), and persists even in settings where the gap between the best and worst equilibrium is large.

We note that in these results, the hardness is not simply borrowed from welfare maximization;\footnote{Welfare maximization refers to the optimization problem of finding an action profile (not necessarily an equilibrium) with maximum possible welfare.} even if the underlying game admits a nontrivial multiplicative approximation for welfare maximization, the problem of determining a CCE with welfare arbitrarily better than the worst CCE remains hard. Another relevant observation is that there always exists an optimal CE/CCE with support size polynomial in the number of players and the number of actions per player.\footnote{This follows from the fact that CE/CCE are defined by a polynomial number of linear constraints. That is, the set of CE/CCE form a polytope that is defined by a polynomial number of linear inequalities. An optimal CE/CCE is an extreme point of this polytope, and hence its support size is polynomially bounded.} Therefore, the fact that in multiplayer games there might exist CE/CCE with exponentially large support size does not, in and of itself, account for the complexity of this problem. 

We complement these hardness results by developing an algorithmic framework for computing an $\varepsilon$-approximate CE with welfare that is additively $\varepsilon$ close to the optimal. This framework establishes a sufficient condition under which the above-mentioned complexity barriers can be circumvented. In particular, we show that if in a given game we can efficiently obtain an \emph{additive} approximation for a \emph{modified-welfare maximization problem}, then we can efficiently compute an approximate CE with high welfare. The modified welfare under consideration can be thought of as a Lagrangian corresponding to the equilibrium constraints (see Definition~\ref{def:mod-wel}), and the modified-welfare maximization problem entails  finding an action profile that maximizes this modified welfare. Note that even if welfare (specified by the given utilities) is nonnegative, modified welfare can be negative for certain action profiles. This notably differentiates welfare maximization and modified-welfare maximization, and provides an idea of the technical challenges that one faces when approximating the modified-welfare maximization problem. (Recall that typical multiplicative-approximation techniques cannot handle negative quantities.) Hence, in a given game, the problem of (nontrivially) approximating the modified-welfare maximization problem can be hard, even if the game admits a nontrivial multiplicative-approximation for welfare maximization.  

Further, we instantiate this algorithmic framework to compute high-welfare approximate CE in \emph{aggregative games}. These are games wherein the utility of each player is a function of her own action and an aggregate (a constant-dimensional summary vector) of all players' actions; see Section~\ref{sect:agg-game} for a formal definition. Aggregative games encompass settings like Cournot oligopolies, Bertrand competitions, weighted congestion games, and anonymous games~\cite{jensen2010aggregative, acemoglu2013aggregate, babichenko2013best, cummings2014privacy}. We develop an efficient additive-approximation algorithm for the modified-welfare maximization problem in aggregative games. Therefore, via the above-mentioned framework, we show how to efficiently compute a high-welfare approximate CE in aggregative games. \\

\paragraph{Related Work}

Papadimitriou and Roughgarden~\cite{PR} showed that the problem of computing an exactly optimal CE is \rm{NP}-hard for many relevant classes of multiplayer games, including congestion games, graphical games, polymatrix games, local effect games, and scheduling games. 
Specific instances in which the hardness result of~\cite{PR} can be completely circumvented, i.e., settings where an \emph{exactly optimal} CE can be efficiently computed, were identified by Jiang and Leyton-Brown~\cite{JLB}. The results in~\cite{PR} and~\cite{JLB} leave open the question of efficiently computing a CE with \emph{near-optimal} welfare, i.e., the question of \emph{approximating} the optimization problem under consideration. The complexity of this approximation is the focus of our work. 

Our main result is negative. In order to prove a positive result for the specific case of computing near-optimal approximate CE in \emph{aggregative games}, we consider a \emph{modified-welfare maximization problem} (MWMP); see Section~\ref{sect:blackwell} for a formal definition. Jiang and Leyton Brown~\cite{JLB} consider classes of games in which the MWMP can be solved optimally, and use the ellipsoid method to find an optimal CE. In our setting, exactly solving the MWMP is not computationally feasible (and hence the framework of~\cite{JLB} cannot be applied to aggregative games\footnote{Knapsack reduces to the problem of welfare maximization in aggregative games.}), but we show that an additive approximation of MWMP suffices to find a near-optimal approximate CE. This entails developing a new algorithm that does not rely on the ellipsoid method.

There is prior work~\cite{charikar2008online,kleinberg2009multiplicative,balcan2013circumventing} on dynamics that quickly converge to high-welfare CCE in isolated, specific classes of games, such as fair cost sharing games; our results show that it is unlikely that such results can be significantly generalized. Marden et al.~\cite{marden2012achieving} develop dynamics that {\em eventually} converge to Pareto-optimal CCE (see also~\cite{song2011optimal}); these works do not establish polynomial rate of convergence for the proposed dynamics.

%% file: notation.tex
\section{Notation}
\label{sect:notation}
In this paper we consider 
games with $n$ players and $m$ actions per player. We use $A_p$ to denote the set of actions available to the $p$th player and $A$ to denote the set of action profiles,  $A := \prod_p A_p$.  We write $u_p : A \rightarrow [0,1]$ for the (normalized) utility of player $p$, and $w: A \rightarrow \mathbb{R}$ is the welfare of an action profile, $w(a) := \sum_{p=1}^n u_p(a)$.\footnote{When there are multiple games under consideration within a single proof, we annotate the $w$ to indicate to which game it pertains.} For an action profile $a \in A$, let $a_{-p}$ denote the profile of actions chosen by players other than $p$. With $A_{-p} :=\prod_{q\neq p} A_q$, we have $a_{-p} \in A_{-p}$.

As is typical in the literature, we say that a game is succinct if it has an efficient representation. Formally, an $n$-player $m$-action game is said to be \emph{succinct} if the player utilities are completely specified via a polynomial-sized string from an input set $I$. Specifically, for a succinct game, there exists a polynomial (in $n$ and $m$) time algorithm $U$ that, given a representation $z \in I$ along with a player $p$ and action profile $a$, returns the utility $u_p(a) = U(z, p, a)$. The game is denoted by $\Gamma(z)$. Many important classes of multi-player games are succinct, e.g., symmetric games, anonymous games, local effect games, congestion games, polymatrix games, graphical games, and network design. This paper is focused on succinct games, since this lets us formally treat settings in which the input, i.e., the utilities in the game, can be efficiently represented. Note that the hardness question becomes moot if we consider the normal form representation of an $n$-player $m$-action game as our ``input,'' since in this case the input itself is exponentially large in $n$ and $m$. Our hardness results imply the intractability of determining high-welfare CCE in games wherein the underlying utilities are specified through a black box. 

We denote the set of probability distributions over a set $B$  by $\Delta(B)$. Given a distribution $x$ over the action profiles $A$, i.e., $x \in \Delta(A)$, we use $u_p(x)$ for the expected utility of player $p$ under distribution $x$. Similarly, we write $w(x)$ to denote the expected welfare under $x$.

\begin{definition}[Correlated Equilibrium]
\label{def:ce}
A probability distribution $x \in \Delta(A)$ is said to be a correlated equilibrium if for every player $p$ and every actions $i,j \in A_p$ we have
\begin{align*}
\sum_{a_{-p} \in A_{-p}}  [u_p(j,a_{-p}) - u_p(i,a_{-p}) ] x(i,a_{-p}) \leq 0,
\end{align*}
where $(i,a_{-p})$ denotes an action profile in which player $p$ plays action $i$ and the other players play $a_{-p}$.
\end{definition}

\begin{definition}[Coarse Correlated Equilibrium]
\label{def:cce}
A probability distribution $x \in \Delta(A)$ is said to be a coarse correlated equilibrium if for every player $p$ and every action $j \in A_p$ we have
\begin{align*}
\sum_{a \in A}  [u_p(j,a_{-p}) - u_p(a) ] x(a) \leq 0,
\end{align*}
where $(j,a_{-p})$ denotes an action profile in which player $p$ plays action $j$ and the other players play $a_{-p}$.
\end{definition}

Along these lines, the definition of an approximate correlated equilibrium is as follows:  
\begin{definition}[$\varepsilon$-Correlated Equilibrium]
\label{def:eps-CE}
A probability distribution $x \in \Delta(A)$ is said to be an $\varepsilon$-correlated equilibrium if for every player $p$ and every actions $i,j \in A_p$ we have
\begin{align*}
\sum_{a_{-p} \in A_{-p}}  [u_p(j,a_{-p}) - u_p(i,a_{-p}) ] x(i,a_{-p}) \leq \varepsilon.
\end{align*}
\end{definition}

Finally, we define $\varepsilon$-coarse correlated equilibrium. 

\begin{definition}[$\varepsilon$-Coarse Correlated Equilibrium]
\label{def:approxCCE}
A probability distribution $x \in \Delta(A)$ is said to be an $\varepsilon$-coarse correlated equilibrium if for every player $p$ and every action $i \in A_p$ we have
\begin{align*}
\sum_{a \in A}  [u_p(i,a_{-p}) - u_p(a) ] x(a) \leq \varepsilon.
\end{align*}
\end{definition}

%% file: hardness.tex
\section{Hardness Results}
In this section we show that, given a succinct game, it is $\rm{NP}$-hard to compute a CCE with welfare strictly better than the lowest-welfare CCE. In particular, we develop a reduction that shows that the following decision problem is $\rm{NP}$-hard.

\begin{definition}[\OP\hspace*{-4pt}]
Let $\Gamma$ be an $n$-player $m$-action game with a succinct representation. \OP is defined to be the problem of determining whether $\Gamma$ admits a coarse correlated equilibrium $x$ such that $w(x) > w(x')$. Here $x'$ denotes the worst CCE of $\Gamma$, in terms of social welfare $w$. 
\end{definition}
 
\hfill 

The hardness of \OP implies that, under standard complexity-theoretic assumptions, any nontrivial approximation of the the optimization problem (\ref{opt-cce}) is impossible. Specifically,  let $x^*$ denote an optimal CCE of a game (i.e, $x^*$ is an optimal solution of the optimization problem (\ref{opt-cce})) and $x'$ be a CCE with minimum possible welfare. Write $\beta:= w(x') / w(x^*)$, i.e., the ratio of the welfare of the worst CCE to that of the best CCE. In games in which a CCE can be computed efficiently, an efficient $\beta$-approximate solution of (\ref{opt-cce}) is direct;  we can simply return an arbitrary CCE.  The hardness of \OP implies that no approximation factor better than $\beta$ can be achieved in general games. A proof of the $\rm{NP}$-hardness of \OP is detailed below. 

\begin{align}
\max \ \ & \ \ \sum_{p=1}^n u_p(x) \nonumber \\
\textrm{subject to} \ \ & \ \ x \textrm{ is a CCE} \label{opt-cce}
\end{align}

\begin{theorem}
\label{thm:op}
\OP is $\rm{NP}$-hard in succinct multiplayer games. 
\end{theorem}

\begin{proof}
We start with a succinct game $G$ from a class of games in which computing a welfare-maximizing action profile is $\rm{NP}$-hard. Multiple examples of such classes of games are given in~\cite{PR}. We reduce the problem of determining an optimal (welfare maximizing) action profile in $G$ to solving \OP in a modified succinct game $G'$. When $G$ is an $n$-player $m$-action succinct game, we construct a modified game $G'$ by providing an additional action, $b_p$, to each player $p \in [n]$. $G'$ is therefore an $n$-player $(m+1)$-action game. 

Let $A$ denote the set of action profiles of game $G$; similarly, let $A'$ be the action profiles of $G'$. Let $u_p: A \rightarrow [0,1]$ and $u'_p: A' \rightarrow [0,1] $ denote the utility of a player $p$ in $G$ and $G'$, respectively. Along these lines, let $w(\cdot)$ and $w'(\cdot)$ represent the welfare of action profiles in $G$ and $G'$, respectively. 
Note that for every action profile $a' \in A' \setminus A$ there exists at least one player $p$ who is playing the augmented action $b_p$, i.e., $a'_p = b_p$. 

Specifically, we start with the following \rm{NP}-hard problem: given succinct game $G$ and parameter $\OPT$, determine if there exists an action profile $a \in A$ such that $w(a) \geq \OPT$.\footnote{Note that here we are considering an \rm{NP}-hard {\em decision} problem and, hence, parameter $\OPT$ is part of the input.}  The utilities $u_p$ (and hence also $w$) are given as succinct input. Using them we define $u'_p$ as follows:

\begin{enumerate}
\item For every action profile $a \in A$, $u'_p(a) := w(a)/n$. In other words, on action profiles that belong to the original game we construct an identical-interest game. 
\item  For every action profile $a' \in A' \setminus A$ such that in $a'$ there is \emph{exactly} one player $p$ who is playing the augmented action $b_p$ (i.e., $a'_p = b_p$ for exactly one player $p$ and $a'_q \neq b_q$ for all $q \neq p$), set $u'_p(a') := \OPT/n $ and $u'_q(a') := 0$ for all $q \neq p$. 
\item For action profiles $a' \in A' \setminus A$ in which more than one player is playing the augmented action $b_p$, we set 
\[u'_p(a') = \left\{
  \begin{array}{ll}
    \varepsilon/n &  \quad \textrm{ if } a'_p = b_p \\
    0 &  \quad \textrm{ otherwise }
  \end{array}
\right.
\]
Here we select $\varepsilon$ to satisfy: $\OPT > \varepsilon \geq \OPT/n$. 
\end{enumerate}

Note that $G'$ is a succinct game. Specifically, if game $G$ is succinct then, by definition, we have a polynomial-size specification $z$ for $G$. In addition, there exists an algorithm $U$ that takes as input $z$, $p \in [n]$, and $a \in A$, and computes the utility of player $p$ at any action profile $a$, i.e., $u_p(a)$, in polynomial time. Now to obtain a succinct representation for $G'$ we can use $z$ and $U$ (as a subroutine) and compute utilities $u'_p$ for any player $p$ and action profile $a$ in polynomial time.

Say $b$ denotes the action profile wherein each player is playing the augmented action, $b:=(b_1, b_2, \ldots, b_n)$. The definition of $u'_p$ implies that $w'(b) := \sum_{p=1}^n u'_p(b) = \sum_{p=1}^n \varepsilon/n = \varepsilon $.

We will prove that there exists an action profile $a \in A$ (i.e., an action profile in game $G$) with $w(a) \geq \OPT$  \emph{iff} there exists a CCE $x$ in $G'$ that satisfies  $w'(x) > w'(b)$. This shows that determining if there exists a CCE $x$ such that $w'(x) > w'(b)$ is $\rm{NP}$-hard. 

To complete the hardness proof for \OP we will show that action profile $b$ is a pure Nash equilibrium (and, therefore, a CCE), and that no other CCE in $G'$ has welfare $w'$ less than $b$.

Suppose $a$ is an optimal action profile in game $G$, i.e., $a \in A $ and $w(a) \geq \OPT$. Then $a$ is in fact a pure Nash equilibrium in $G'$. This follows from the fact that $u'_p(a) = w(a)/n \geq \OPT/n$ ($G'$ is identical interest on $a \in A$); hence (i) for any possible deviation $\hat{a}_p \neq b_p $ for player $p$ we have $u'_p(a) = w(a)/n \geq w(\hat{a}_p, a_{-p})/n = u'_p(\hat{a}_p, a_{-p})$. The first inequality holds since $a$ is an optimal action profile in $G$; (ii) for deviation $b_p$, note that $u'_p(a) \geq \OPT/n = u'_p(b_p, a_{-p})$. Therefore, no player can benefit (increase $u'_p$) by unilaterally deviating from $a$, thereby proving that $a$ is a pure Nash equilibrium in $G'$. Overall, we get that if there exists an action profile $a \in A$ with $w(a) \geq \OPT$  then there exists a CCE $x$ (in particular, an optimal action profile $a$ itself) in $G'$ that satisfies  $w'(x) > w'(b)$. Recall that $w'(a) \geq \OPT > \varepsilon = w'(b)$. 

It remains to show that if there exists a CCE $x$ such that $w'(x) > w'(b)$ then there exists an action profile $a \in A$ with $w(a) \geq \OPT$. We will consider the set of action profiles in the support of $x$ that are also contained in $A$, i.e., $\textrm{Supp}(x) \cap A$. A useful observation is that for all $a' \in A' \setminus A$ the welfare $w'$ satisfies: $w'(a') \leq w'(b)$ (recall, $\varepsilon \geq \OPT/n$). 
This implies that $\textrm{Supp}(x) \cap A \neq \phi$; otherwise, we would have $w'(x) \leq w'(b)$. Write $\pi>0$ to denote the probability mass of $x$ on the set $\textrm{Supp}(x) \cap A$; specifically, $\pi := \sum_{a \in A} x(a)$.
 
Since $x$ is a CCE, deviating to $b_p$ could not increase any player $p$'s expected utility: 
\begin{align*}
w'(x) & = \E_{a \sim x} [ u'_p(a) ] \\ 
& \geq \E_{a \sim x} [u'_p(b_p, a_{-p})].
\end{align*}

We can rewrite the above inequality as follows: $\E_{a \sim x} [ u'_p(a) - u'_p(b_p, a_{-p})] \geq 0$. Next we expand in terms of conditional expectation

\begin{align}
\E_{a' \sim x} \left[ u'_p(a') - u'_p(b_p, a'_{-p})  \mid a \in A' \setminus A \right] \cdot (1-\pi) \ + \ \E_{a \sim x} \left[u'_p(a) - u'_p(b_p, a_{-p}) \mid a \in A \right] \cdot \pi & \geq 0. \label{ineq:cexp}
\end{align}

Note that for any action profile $a' \in A' \setminus A$ and each player $p$ we have \begin{align} u'_p(a') - u'_p(b_p, a'_{-p}) \leq 0.\end{align} Either $a'_p = b_p$, in which case $u'_p(a') - u'_p(b_p, a'_{-p}) = 0$; otherwise, $a'_p \neq b_p$ and then $u'_p(a') = 0 < u'_p(b_p, a'_{-p}) $.

This implies that the term $\E_{a' \sim x} \left[ u'_p(a') - u'_p(b_p, a'_{-p})  \mid a \in A' \setminus A \right] $ in inequality (\ref{ineq:cexp}) is non-positive for every player $p$. Therefore, the second term in  (\ref{ineq:cexp}), $\E_{a \sim x} \left[u'_p(a) - u'_p(b_p, a_{-p}) \mid a \in A \right]$, must be non-negative for every player $p$.  Summing the second term over all players we get: 
\begin{align}
\ \E_{a \sim x} \left[ \sum_p \left( u'_p(a) - u'_p(b_p, a_{-p}) \right)\mid a \in A \right] \cdot \pi & \geq 0.  \label{ineq:fexp}
\end{align}


Recall that $\pi > 0$, i.e., there exists an action profile $a \in A$ such that $x(a) > 0$. Therefore, inequality (\ref{ineq:fexp}) and the probabilistic method imply that there exists an action profile $a \in A $ such that $\sum_p \left( u'_p(a) - u'_p(b_p, a_{-p}) \right) \geq 0$. Since $a \in A$, $u'_p(b_p, a_{-p}) = \OPT/n$ for all $p$. Hence, $\sum_p u'_p(a) \geq \sum_p \OPT/n = \OPT$.

Thus, the existence of a CCE $x$ in $G'$ such that $w'(x) > w'(b)$ implies that there exists an action profile $a \in A$  with $w(a) \geq \OPT$.

To complete the proof, we need to show that $b$ is a pure Nash equilibrium and that no other CCE in $G'$ has welfare $w'$ less than $b$. The first part of this claim is direct. To prove the second part, suppose by way of contradiction that there existed a CCE $x'$ in $G'$ such that $w'(x') < w'(b)$. Therefore there would exist a player $p$ such that  
\begin{align}
\E_{a' \sim x'} [ u'_p(a') ] & < w'(b)/n \nonumber \\
& = \varepsilon/n. \label{ineq:switch}
\end{align}
But, note that for any $a'_{-p} \in A'_{-p}$ we have $u'_p(b_p, a'_{-p}) \geq \varepsilon/n$. This observation along with inequality (\ref{ineq:switch}) implies that $p$ would strictly benefit by unilaterally deviating to $b_p$. Therefore, $x'$ cannot be a CCE. This completes the proof. 
\end{proof}


\noindent 
{\bf Remark:} The proof of Theorem~\ref{thm:op} can be directly adopted to establish hardness for CE as well. In particular, the fact that any CE $x$ satisfies the inequalities that define a CCE (see Definitions~\ref{def:ce} and~\ref{def:cce}) can be used in the previous proof to show that it is \rm{NP}-hard to determine a CE with welfare strictly better than the worst possible CE. 

In addition, we show below that the reduction given in the proof of Theorem~\ref{thm:op} establishes a hardness result for the egalitarian objective as well. 

\begin{theorem}
\label{thm:eo}
In an $n$-player, $m$-action succinct game it is {\rm NP}-hard to determine if there exists a coarse correlated equilibrium $x$ that satisfies $\min_p u'_p(x) > \min_p u'_p(x')$, where $u'_p$ denotes the utility of player $p$ in the given game and $x'$ is the worst equilibrium with respect to the egalitarian objective, i.e., $x' \in \argmin_{x'' \in \textrm{CCE }} \{ \min_p u'_p(x'') \}$. 
\end{theorem}

\begin{proof}[Sketch] Here we use the same notation as in the proof of Theorem~\ref{thm:op}. Also, as in the previous proof, we obtain a reduction from the following \rm{NP}-hard problem: given succinct game $G$, determine if there exists an action profile $a \in A$ such that $w(a) \geq \OPT$. 

Note that the action profile $b$ in the constructed game $G'$ is the worst equilibrium with respect to the egalitarian objective, i.e., $ b \in \argmin_{x'' \in \textrm{CCE }} \{ \min_p u'_p(x'') \}$. We can establish this fact by contradiction. In particular, if there existed a CCE $x'$ such that $\min_p u'_p(x') < \min_p u'_p(b) = \varepsilon/ n$, then the player $p$ that obtains the minimum utility under $x'$ could benefit by unilaterally deviating to $b_p$, contradicting the assumption that $x'$ is a CCE.

To prove this theorem we show that the original game $G$ has an action profile $a$ with $w(a) \geq \OPT$ iff there exists a CCE $x$ such that $\min_p u'_p(x) > \min_p u'_p(b)$. The forward direction follows from the fact that an optimal action profile $a$ with welfare at least $\OPT$ is a pure Nash equilibrium in $G'$. To establish the reverse direction we note that $u'_p(b) = \varepsilon/n$ for all $p$. Hence if a CCE $x$ satisfies $\min_p u'_p(x) > \min_p u'_p(b)$, then its welfare $w'(x)$ is strictly greater than $\varepsilon$. In other words, $w'(x) > \varepsilon = w'(b)$. But, as shown in the previous proof, this strict inequality suffices to establish the existence of an action profile for which $w(a) \geq \OPT$. Hence, we get the desired claim. 
\end{proof}

\noindent 
{\bf Remark:} The reduction detailed above also proves that there does not exist a polynomial-time algorithm that computes a Pareto-efficient CCE, unless \rm{P=NP}. We can establish this result by noting that a polynomial time algorithm, say $\A$, that computes any Pareto-efficient CCE can be used to determine whether there exists an action profile $a$ that satisfies $w(a) \geq \OPT$; as before, this suffices to prove the hardness result. 

If $\A$ returns $b$ as a Pareto-efficient equilibrium then we know that there does not exist an action profile $a$ such that $w(a) \geq \OPT$, since such an action profile would Pareto dominate $b$ in $G'$:  $u'_p(a) > u'_p(b) $ for all $p$. Also, note that if $\A$ returns a Pareto-efficient CCE $x$ such that $u'_p(x) = u'_p(b)$ for all $p$, then again we get that $b$ is Pareto-efficient. So this case is subsumed in the first one. Recall that every CCE $x$ of $G'$ satisfies $u'_p(x) \geq u'_p(b)$. Therefore, the final case entails $\A$ returning a CCE $x$ such that for some $p$ we have $u'_p(x) > u'_p(b) $. Hence, we get that $w'(x) > w'(b)$, which again implies the existence of an action profile $a$ with welfare $w(a) \geq \OPT$. \\

\noindent 

{\bf Remark:} Theorems~\ref{thm:op} and~\ref{thm:eo} hold for potential games. This follows from the fact that the reduction used in the proof of these theorems in fact gives us a potential game. Specifically, a potential function $\phi$ for the constructed game $G'$ is as follows: 
\begin{enumerate}
\item $\phi(a) := w(a)/n$ for all $a \in A$.
\item For all action profiles $a \in A' \setminus A$ (i.e., in $a$ at least one player is playing is playing its augmented action $b_p$), we set 
\[ \phi(a) := \frac{\OPT}{n} + \frac{(k-1)\varepsilon}{n}. \] 
Here $k$ is the number of players playing their corresponding augmented action $b_p$ in action profile $a$, $k = \left| \{ p \mid a_p = b_p \} \right|$.

\end{enumerate}

A case analysis shows that $\phi$ is a potential function for $G'$. 
In particular, we will show that the following equality holds for each player $p$ and action profiles $(a_p, a_{-p})$ and $(a'_p, a_{-p})$:
\begin{align}
u'_p(a_p, a_{-p}) - u'_p(a'_p,a_{-p}) & = \phi(a_p, a_{-p}) - \phi(a'_p,a_{-p}) \label{eq:pot}
\end{align}

\begin{itemize}
\item[] Case \rm{I}: Both $(a_p, a_{-p})$ and $(a'_p, a_{-p})$ are action profiles in $A$. Here we have $u'_p(a_p, a_{-p}) = w(a_p, a_{-p})/n = \phi(a_p, a_{-p}) $ and $u'_p(a'_p, a_{-p}) = w(a'_p, a_{-p})/n = \phi(a'_p, a_{-p})$. Hence, in this case (\ref{eq:pot}) holds. \\

\item[] Case \rm{II}: Action profile $(a_p, a_{-p}) \in A$ and  $(a'_p, a_{-p}) \notin A$ (i.e., $a'_p = b_p$). Again, following the definitions of utility $u'_p$ and potential function $\phi$ we get the equality (\ref{eq:pot}):  $u'_p(a_p, a_{-p}) = w(a_p, a_{-p})/n = \phi(a_p, a_{-p}) $ along with $u'_p(a'_p, a_{-p}) = \OPT/n = \phi(a'_p, a_{-p})$. The symmetric case of $(a_p, a_{-p}) \notin A$ and  $(a'_p, a_{-p}) \in A$ is similarly addressed. \\

\item[] Case \rm{III}: Both action profiles $(a_p, a_{-p})$ and $(a'_p, a_{-p})$ are not in $A$. If neither $a_p$ nor $a'_p$ is equal to $b_p$ the utility $u'_p$ is zero under both the action profiles. Also, the number of players playing their respective augmented actions $b_q$ is the same in $(a_p, a_{-p})$ and $(a'_p, a_{-p})$, hence $\phi(a_p, a_{-p}) = \phi(a'_p, a_{-p})$. This enforces equality (\ref{eq:pot}). 

Now we consider the setting in which exactly one of $a_p$ or $a'_p$ is equal to $b_p$; say $a_p = b_p$ (the other possibility (i.e., $a'_p = b_p$) holds by  symmetry). Here, $u'_p(a_p, a_{-p}) =  \varepsilon/n$ and $u'_p(a'_p, a_{-p}) = 0$. Say $k \in [n]$ is the number of players playing their corresponding augmented action in action profile $(a_p, a_{-p})$, then  $\phi(a_p, a_{-p}) =\OPT/n + (k-1) \varepsilon/n $ and $\phi(a'_p, a_{-p}) = \OPT/n + (k-2) \varepsilon/n$. Therefore, again, (\ref{eq:pot}) holds.
\end{itemize}


\subsection{Approximate Coarse Correlated Equilibrium}
This section establishes the hardness of computing an {\em approximate} CCE that has high social welfare. Specifically, we consider the problem of computing a $\frac{1}{2n^3}$-CCE with welfare $(1 + \frac{1}{n})$ times better than the welfare of the worst CCE. Note that there exist regret-based dynamics (c.f~\cite{young2004strategic}) that converge to the set of $\varepsilon$-CCE in time polynomial in $1/\varepsilon$. Therefore, in polynomial time we can compute \emph{a} $\frac{1}{2n^3}$-CCE. But, as the following theorem shows, it is unlikely that we can efficiently find a $\frac{1}{2n^3}$-CCE with any nontrivial welfare guarantee.  

 
Note that in an $n$-player $m$-action game a $\frac{1}{2n^3m}$-CE is guaranteed to be a $\frac{1}{2n^3}$-CCE (see Definitions~\ref{def:eps-CE} and~\ref{def:approxCCE}). Using this fact, one can directly use the proof given in this section to show that, under standard complexity-theoretic assumptions, there does not exist a polynomial time algorithm that determines a $\frac{1}{2n^3m}$-CE with any nontrivial welfare guarantee in succinct multiplayer games. It is worth pointing out that in multiplayer games we can always find {\em a} $\frac{1}{2n^3m}$-CE in polynomial time (c.f~\cite{young2004strategic}). 

\begin{definition}[\AOP]
Let $\Gamma$ be an $n$-player $m$-action succinct game. \AOP \ is defined to be the problem of determining whether there exists a $\frac{1}{2n^3}$-CCE $x$ in $\Gamma$ such that $w(x) \geq (1 + \frac{1}{n}) w(x')$, where $x'$ denotes the worst CCE of $\Gamma$, in terms of social welfare $w$. 
\end{definition}



\begin{theorem}
\label{thm:aop}
In succinct multiplayer games, \AOP \ is $\rm{NP}$-hard under randomized reductions: if \AOP \ admits a polynomial-time algorithm then {\rm NP} admits a polynomial-time randomized algorithm. 
\end{theorem}

\begin{proof}
We will extend the construction presented in the proof of Theorem~\ref{thm:op}. We start with a game $G$ from a class of games in which it is {\rm NP}-hard to compute an action profile with welfare within one of the optimal. That is, in $G$ it is \rm{NP}-hard to compute an action profile $a$ such that $w(a) \geq \max_{a' \in A} w(a')  - 1$; note that this is a fairly modest hardness of approximation requirement. 



Write $\OPT = \max_{a \in A} w(a)$. Below we develop a polynomial-time randomized algorithm that uses an algorithm for \AOP \ to compute an action profile $a$ that satisfies $w(a) \geq \OPT -1$. This establishes the stated claim. 

To find the desired action profile $a$, we need a parameter $\tau$ that satisfies $\tau \in [\OPT -1 , \OPT]$. 
Since the utilities in $G$ are normalized between $0$ and $1$, we have $\OPT \leq n$. Therefore, one of the values in $\{0,1,\ldots, n-1\}$ will give $\tau \in [\OPT -1 , \OPT]$, and we can simply search exhaustively. 


Applying the same transformations as in the proof of Theorem~\ref{thm:op}, we obtain the succinct game $G'$. While setting utilities in $G'$ we use  $\varepsilon = (\tau +1)/n$, where parameter $\tau \in [\OPT- 1, \OPT]$. Therefore, we have $ \frac{\OPT}{n} \leq \varepsilon \leq \frac{\OPT + 1}{n}$. 

We assume that $\OPT \geq 1$, else finding an action profile $a$ such that $w(a) \geq \OPT -1$ is trivial. Also, we can assume that $n \geq 4$; recall that for a constant number of players, an optimal CCE can be computed in polynomial time. The following inequality holds under these assumptions: $\OPT \geq \left( 1 + \frac{1}{n} \right) \frac{\OPT + 1}{n} $.

As before, the action profile $b$ is a pure Nash equilibrium, and in fact is a CCE with minimum social welfare. 


First, note that an optimal action profile $a^* \in \argmax_{a \in A} w(a)$ of $G$ is a pure Nash equilibrium (hence, a $\frac{1}{2n^3}$-CCE) in $G'$. Also, we have $w'(a^*) = w(a^*) = \OPT$. The bound $w'(a^*) \geq \left(1 + \frac{1}{n}\right) w'(b)$ follows from the following chain of inequalities: $\OPT \geq \left(1 + \frac{1}{n} \right) \frac{\OPT +1}{n} \geq \left(1 + \frac{1}{n}\right) \varepsilon = \left(1 + \frac{1}{n}\right) w'(b)$. Thus we get that there exists a $\frac{1}{2n^3}$-CCE with welfare strictly better than $(1+1/n) w'(b)$. This overall ensures that a polynomial-time algorithm for \AOP \ is guaranteed to return a solution. Next we show that any such returned solution can be used to compute an action profile $a$ that satisfies $w(a) \geq \OPT -1$.




The fact that $w'(a) \leq w'(b) $ for all $a \in A' \setminus A$ and the inequality $w'(x) \geq (1 + \frac{1}{n}) w'(b)$ imply that $\sum_{a \in A} w'(a) x(a) \geq \frac{1}{n} w'(b)$. Recall that $w'(b) = \varepsilon  \geq \frac{\OPT}{n}$. Therefore, $\sum_{a \in A} w'(a) x(a) \geq \frac{1}{n^2} \OPT$. Since $w(a) = w'(a)$ for all $a \in A$, we have  $\max_{a \in A} \ w'(a) = \OPT$. Therefore, $\pi := \sum_{a \in A} x(a) \geq \frac{1}{n^2}$. 

Given that $x$ is a $\frac{1}{2n^3}$-approximate CCE, analogous to inequality (\ref{ineq:fexp}) here we have 
\begin{align}
\E_{a \sim x} \left[ \sum_p \left( u'_p(a) - u'_p(b_p, a_{-p}) \right)\mid a \in A \right] \cdot \pi & \geq - \frac{1}{2n^2}.  \label{ineq:fexpa}
\end{align}

Since $\pi \geq \frac{1}{n^2}$, inequality (\ref{ineq:fexpa}) implies $\E_{a \sim x} \left[ \sum_p \left( u'_p(a) - u'_p(b_p, a_{-p}) \right)\mid a \in A \right]  \geq - 1/2$.

For all $a \in A$ and $p \in [n]$, we have $u'_p(b_p, a_{-p}) = \OPT/n$. Therefore, we get the following bound on the conditional expectation $\E_{a \sim x} \left[ \sum_p  u'_p(a) \mid a \in A \right]  \geq \OPT - 1/2$. For all action profiles $\sum_p u_p'(a) = w'(a) \leq \OPT \leq n$. This implies that in the conditional distribution $\Pr_x( a \mid a \in A)$ the probability mass on action profiles that satisfy $w'(a) \geq \OPT -1$ is at least $\frac{1}{2n}$. 

Therefore, with high probability, we can obtain an action profile that satisfies $w'(a) \geq \OPT -1$ by drawing polynomially many independent and identically distributed (i.i.d.) samples from the  conditional distribution $\Pr_x( a \mid a \in A)$. Since $\pi = \sum_{a \in A} x(a) \geq  \frac{1}{n^2}$, we can obtain polynomially many i.i.d.\ samples from the conditional distribution by drawing polynomially many i.i.d.\ samples from $x$. This overall gives us a polynomial-time randomized algorithm to find an action profile that satisfies $w(a) = w'(a) \geq \OPT -1$. Hence, the stated claim follows. 
%
%
 %
\end{proof}

\noindent
In this section we considered approximate CCE with a specific approximation factor, i.e., we established hardness for $\frac{1}{2n^3}$-CCE. This was for  ease of presentation, and in fact hardness of a parameterized version of \AOP \ can be obtained along the lines of the given proof. In particular, we can show that for any $\delta \in \left[ \frac{1}{\textrm{poly}(n)}, 1 \right]$ it is computationally hard to compute a $\frac{\delta}{2n^2}$-CCE with welfare greater than $\left(1 + \delta \right) w(x')$, where, again, $x'$ denotes the worst CCE. 


%% file: algo.tex
\section{Computing Approximate Correlated Equilibria with Near-Optimal Welfare}
\label{sect:blackwell}


In this section, we develop an algorithmic framework for computing an $\varepsilon$-CE with welfare additively $\varepsilon$ close to the optimal. The ideas presented in the section can be easily modified to find an $\varepsilon$-CCE with welfare additively $\varepsilon$ close to the optimal CCE.\footnote{In order to find an approximate CCE with near-optimal welfare we can
define a different regret vector than the one under consideration in this section, whose components are equal to the regret
terms that appear in the definition of a CCE. Note that  the regret vector
for the CCE case is $nm+1$ dimensional.}

Our framework is based on a novel extension of Blackwell's condition, which is used in the analysis of no-regret algorithms (see, e.g.,~\cite{young2004strategic}). The idea here is to define, for each action profile $a \in A$, a vector $r(a)$ whose components list the regret of each player at action profile $a$. Specifically, for each player $p \in [n]$ and action $j \in [m]$ there is component in $r(a)$ that is equal to $u_p(j,a_{-p}) - u_p(a) $; note that this quantity is the the regret of player $p$ at action profile $a$ with respect to deviation $j$. The regret vector $r(a)$ has an additional component that is equal to the the difference between the optimal welfare and the welfare of action profile $a$. 

Intuitively, the components of $r(a)$ are defined to ensure that $x^*$ is an optimal CE if and only if the following component-wise inequalities hold: $\E_{a \sim x^*} [ r(a) ] \leq 0$. Moreover, to find the desired approximate CE it suffices to determine a distribution $x \in \Delta(A)$ that satisfies $\E_{a \sim x} [ r(a) ] \leq \varepsilon$.  Using an extension of Blackwell's condition (see inequality (\ref{ineq:bwcond})), we develop an algorithm for finding such a distribution $x$. In particular, via a gradient-descent like argument, we show that action profiles satisfying the extended Blackwell condition can be used to determine the desired approximate CE $x$; see proof of Theorem~\ref{thm:graddes} for details. 

It turns out that finding an action profile that satisfies the extended Blackwell condition corresponds to computing an additive approximation of a modified-welfare maximization problem (see Definition~\ref{def:mod-wel}). This overall gives us an algorithmic framework that reduces the problem of determining an approximate CE with near optimal welfare to the problem of additive approximating a modified-welfare maximization problem. We instantiate this framework in the context of \emph{aggregative games} in the next section. 

Formally, we begin by defining a $d = nm(m-1) + 1$ dimensional regret vector $r(a)$ for each action profile $a \in A$. The first $nm(m-1)$ components of $r(a)$ are indexed by triples $(p,i,j)$ for player $p \in [n]$ and distinct actions $i,j \in [m]$. The $(p,i,j)$th component of $r(a)$ is equal to $u_p(j,a_{-p}) - u_p(a) $ if $a_p = i$, and is zero otherwise. That is, the $(p,i,j)$th component is the regret that player $p$ experiences at action profile $a$ by not playing action $j$. The last ($d$th) component of $r(a)$ is equal to $w^* - w(a)$. Here $w^*$ denotes the optimal welfare over the set of correlated equilibria, i.e., $w^* := \max \{  w(x) \mid x \textrm{ is a correlated equilibrium}\}$.

Write $x^*$ to denote the welfare-optimal CE, i.e., $w^* = w(x^*)$. Note that for $x^*$ we have that $\E_{a \sim x^*} [ r(a) ] \leq 0$ holds component-wise . 
Now, a useful observation is that for any scaling vector $y \in \mathbb{R}_{+}^d$ with nonnegative components, we have $\E_{a \sim x^*} [ y^T r(a) ] \leq 0$. Via the probabilistic method, we get that for any $y \in \mathbb{R}_{+}^d$ there exists an action profile $a^*$ such that 
\begin{align}
\label{ineq:bwcond}
y^T r(a^*) \leq 0
\end{align}

Inequality (\ref{ineq:bwcond}) can be thought of as an extension of Blackwell's condition. 

This inequality leads us to the objective of maximizing a modified welfare function that is defined as follows.
\begin{definition}[Modified Welfare]
\label{def:mod-wel}
Given scaling vector $y \in \mathbb{R}_{+}^d$ (where the first $nm(m-1)$ components of $y$ are indexed by $(p,i,j)$ for player $p \in [n]$ and actions $i,j \in [m]$ and we refer to the last component of $y$ as $y_d$), we define modified utilities $\tilde{u}^y_p$ and modified welfare $\tilde{w}^y$ as follows:
\begin{align}
\tilde{u}^y_p(a) & := y_d u_p(a) + \sum_{j \in A_p} y_{(p, a_p, j)} (u_p(a) -  u_p(j, a_{-p})) \\
\tilde{w}^y(a) & := \sum_p \tilde{u}^y_p(a).
\end{align}
\end{definition}
%
%
For ease of presentation, when $y$ is clear from context we will drop it from the superscript of $\tilde{u}_p^y$ and $\tilde{w}^y$.

\begin{definition}[Modified-Welfare Maximization Problem]
\label{def:wel-max}
Given a multi-player game and vector $y \in  \mathbb{R}_{+}^d$, the modified-welfare maximization problem (MWMP) is to compute an action profile $a$ of the game that maximizes modified welfare $\tilde{w}^y$, i.e., the objective is to obtain $\argmax_{a \in A} \tilde{w}^y(a)$.
\end{definition} 

Note that for any vector $y \in \mathbb{R}_{+}^d$ and any action profile $a$ we have $y^T r(a) = y_d w^* - \tilde{w}^y(a)  $. As argued above, for any vector $y$ with non-negative components there exists an action profile $a^*$ that satisfies (\ref{ineq:bwcond}). In particular, $a^*$ satisfies $\tilde{w}^y(a^*) \geq y_d w^*$. Therefore, given a vector $y$, we can compute an action profile that satisfies (\ref{ineq:bwcond}) by solving an instance of MWMP specified via $y$.  Moreover, an $\alpha$-\emph{additive} approximation of MWMP is guaranteed to produce an action profile that satisfies $y^T r(a) \leq \alpha$.

Below we show that (additively) approximating this welfare maximization problem is sufficient to obtain an approximate CE with near-optimal welfare. The hardness result established earlier (see Theorem~\ref{thm:aop}) implies that MWMP cannot be efficiently approximated in general succinct games. However, it is possible for us to approximate MWMP in specific classes of games; in particular, the next subsection details an efficient algorithm to approximate MWMP in \emph{aggregative games}.

Specifically, given a game and vector $y \in \mathbb{R}_+^d$, write $\M(y)$ to denote an $O\left( \frac{\varepsilon^4}{ n^4m^8}  \right)$-additive approximation for MWMP with respect to the specified $y$. Here, $\varepsilon$ is an approximation parameter. Note that an additive approximation $a = \M(y)$ satisfies  $y^T r(a) \leq O \left( \frac{\varepsilon^4}{ n^4m^8} \right)$.


Our algorithm, $\A$, for computing an approximate CE is given below. $\A$ requires access to an additive approximation $M(y)$ for polynomially many $y$s. Note that the $y$s considered during $\A$'s execution satisfy $ y \in [0,n]^d$. 

\begin{algorithm}{{\bf Given:} an algorithm for computing additive approximation $\M(y)$ in an $n$-player $m$-action game; {\bf Return:} $\varepsilon$-correlated  equilibrium of the game with welfare at least $w^* - \varepsilon$.}
\caption*{Algorithm for computing $\varepsilon$-correlated equilibrium with near-optimal welfare}
  \begin{algorithmic}[1]
   \label{alg:grad-des}
   \STATE Set $a^0$ to be an arbitrary action profile of the game and $N = O\left(\frac{n^2m^4}{\varepsilon^2}\right)$.
   \STATE Let $\mathcal{N}$ denote the negative orthant and $\Pi_\mathcal{N} (v)$ denote the Euclidean projection of vector $v$ onto $\mathcal{N}$.
\STATE Initialize average regret vector $\bar{r}_0 = r(a^0)$.
   \FOR{ $t=1$ to $N$ } 
   \STATE Set $y = \bar{r}_{t-1} - \Pi_\mathcal{N} (\bar{r}_{t-1} )$.
   \COMMENT{Note that the components of $y$ are nonnegative and their magnitude is no more than $n$}
   \STATE Set $a^t = \M(y)$.
   \COMMENT{Note that $a^t$ satisfies $y^T r(a^t) \leq O \left( \frac{\varepsilon^4}{ n^4m^8} \right) = O\left(\frac{1}{N^2}\right) $} \label{step:inprod} 
\STATE Set $\bar{r}_t = \frac{t}{t+1} \bar{r}_{t-1} + \frac{1}{t+1} r(a^t) $. \label{step:avgrg}
   \ENDFOR
   \STATE Return the empirical distribution over the multiset $\{a^0, a^1, a^2, \ldots, a^N\}$.
    \end{algorithmic}
\end{algorithm}

\begin{theorem}
\label{thm:graddes}
For a given $n$-player $m$-action game, algorithm $\A$ computes an $\varepsilon$-correlated equilibrium with welfare at least $w^* -\varepsilon$. Here  $w^*$ denotes the optimal welfare over the set of correlated equilibria of the given game. Moreover, if in the given game additive approximations $\M(y)$ for $y \in [0,n]^d$ can be computed in polynomial (in $n$, $m$, and $1/\varepsilon$) time, then $\A$ runs in polynomial time as well.
\end{theorem}

\begin{proof}
First we establish the stated running-time bound for Algorithm $\A$. Note that $\A$ iterates $N = O\left(\frac{n^2m^4}{\varepsilon^2}\right)  $ times. Therefore, if additive approximations $\M(y)$ can be computed in polynomial time, then $\A$ runs in polynomial time as well. 

Next we establish that $\A$ computes an approximate correlated equilibrium with high welfare. Write $x$ to denote the distribution returned by $\A$, i.e., $x$ is the empirical distribution over the multiset of action profiles $\{a^0, a^1, a^2, \ldots, a^N\}$. We will show that $x$ satisfies $\E_{a \sim x} [ r(a) ] \leq \varepsilon$, component-wise. This inequality and the definition of regret vector $r(a)$ imply that $x$ is an $\varepsilon$-correlated equilibrium with welfare at least $w^* - \varepsilon$. 

Note that Step (\ref{step:avgrg}) of algorithm $\A$ ensures that $\bar{r}_N = \sum_{t=1}^N \frac{1}{N} r(a^t) = \E_{a \sim x} [ r(a) ] $. Here, the second equality follows from the fact that $x$ is the empirical distribution over action profiles $a^1, a^2, \ldots, a^N$.

Write $d(r, \mathcal{N})$ to denote the Euclidean distance between vector $r$ and the negative orthant $\mathcal{N}$. The proof proceeds by showing that $d(\bar{r}_N, \mathcal{N})$  is no more than $\varepsilon$. This implies that component-wise $\bar{r}_N$ is no more than $\varepsilon$, and hence we get the desired claim $\E_{a \sim x} [ r(a) ] \leq \varepsilon$. Recall that   $\bar{r}_{t-1}$ denotes the average regret vector considered in the $(t-1)$th iteration of the algorithm and $\Pi_\mathcal{N}( \bar{r}_{t-1})$ denotes the Euclidean projection of this vector onto the negative orthant. The vector $\Pi_\mathcal{N}( \bar{r}_{t-1})$ is found by replacing the positive components of $\bar{r}_{t-1}$ by $0$, i.e., the $i$th component of the projection $(\Pi_\mathcal{N}( \bar{r}_{t-1}))_i $ is equal to $\min\{ 0, (\bar{r}_{t-1})_i \}$. We bound the Euclidean distance of $\bar{r}_{t}$ from the negative orthant as follows:
\begin{align}
d^2(\bar{r}_{t}, \mathcal{N}) & \leq  d^2(\bar{r}_{t} , \Pi_\mathcal{N} (\bar{r}_{t-1}))  \nonumber \\
& =  \left\| \frac{t}{t+1}\bar{r}_{t-1} + \frac{1}{t+1} r(a^t) - \Pi_\mathcal{N} (\bar{r}_{t-1})\right\|_2^2  \nonumber \\
& = \left(\frac{t}{t+1}\right)^2 \| \bar{r}_{t-1} - \Pi_\mathcal{N} (\bar{r}_{t-1}) \|_2^2 + \left(\frac{1}{t+1}\right)^2 \| r(a^t) - \Pi_\mathcal{N} (\bar{r}_{t-1}) \|_2^2 \nonumber \\ & \qquad + \frac{2t}{t+1} \left(\bar{r}_{t-1} - \Pi_\mathcal{N} (\bar{r}_{t-1})  \right)^T \left(r(a^t) - \Pi_\mathcal{N} (\bar{r}_{t-1})  \right)   \label{eq:interim}
\end{align}

Next we bound the terms on the right-hand side of equality (\ref{eq:interim}).  The fact that the utilities of the players are between $0$ and $1$ implies that for any action profile $a$ the regret vector satisfies $\| r(a) \|_2^2 \leq 2 n^2m^4$. Also, $\| \bar{r}_{t-1} \|_2^2 \leq 2n^2m^4$, since $\bar{r}_{t-1}$ is an average of regret vectors. Therefore, using the triangle inequality, we get the following bound for the second term in (\ref{eq:interim}), $\left(\frac{1}{t+1}\right)^2 \| r(a^t) - \Pi_\mathcal{N} (\bar{r}_{t-1}) \|_2^2  \leq \left(\frac{2 nm^2}{t+1}\right)^2  $.

Step (\ref{step:inprod}) ensures that $\left(\bar{r}_{t-1} - \Pi_\mathcal{N} (\bar{r}_{t-1})  \right)^T r(a^t)$  is no more than $O \left( \frac{1}{N^2} \right)$. In addition, note that  the nonzero components of vector $\bar{r}_{t-1} - \Pi_\mathcal{N} (\bar{r}_{t-1}) $ are the positive components of vector $\bar{r}_{t-1}$, and on the other hand the nonzero components of vector $ \Pi_\mathcal{N} (\bar{r}_{t-1}) $ are the negative components of $\bar{r}_{t-1}$. Therefore, we have $\left(\bar{r}_{t-1} - \Pi_\mathcal{N} (\bar{r}_{t-1})  \right)^T \Pi_\mathcal{N} (\bar{r}_{t-1}) = 0 $. Overall, we get the following bound on the third term in (\ref{eq:interim}): 
\begin{align*}
\frac{2t}{t+1} \left(\bar{r}_{t-1} - \Pi_\mathcal{N} (\bar{r}_{t-1})  \right)^T \left(r(a^t) - \Pi_\mathcal{N} (\bar{r}_{t-1})  \right)  & \leq O \left( \frac{1}{N^2} \right)  \\
& \leq \frac{1}{(t+1)^2} 
\end{align*}

Using the bounds mentioned above and multiplying equation (\ref{eq:interim}) by $(t+1)^2$ we get 
\begin{align*}
(t+1)^2 d^2(\bar{r}_{t}, \mathcal{N}) & \leq t^2 d^2 (\bar{r}_{t-1}, \mathcal{N}) + O(n^2 m^4).  
\end{align*}
This
leads to a telescoping sum for $1 \leq t \leq N$ that overall gives us
\begin{align}
\label{ineq:bnd}
N^2 d^2(\bar{r}_{N}, \mathcal{N}) & \leq d^2 ( \bar{r}_1, \mathcal{N}) + O(n^2 m^4 N).
\end{align} 

Note that $\| \bar{r}_1 \|_2^2 \leq O(n^2 m^4) $, therefore $d^2 ( \bar{r}_1, \mathcal{N}) \leq O(n^2 m^4) $. Hence, inequality (\ref{ineq:bnd}) gives $N^2 d^2(\bar{r}_{N}, \mathcal{N}) \leq O(n^2 m^4 N)  $. In other words, $ d (\bar{r}_{N}, \mathcal{N}) \leq O(n m^2 / \sqrt{N})$.
Given that $N = O\left(\frac{n^2m^4}{\varepsilon^2}\right)$, we get that the Euclidean distance between $\bar{r}_N$ and the negative orthant is at most $\varepsilon$, i.e., $d (\bar{r}_{N}, \mathcal{N})  \leq \varepsilon$.

As discussed above, the last inequality implies that $\E_{a \sim x} [ r(a) ] \leq \varepsilon$, where $x$ is the distribution returned by the algorithm $\A$. Overall, following the argument outlined above, we get the desired claim. 
\end{proof}

\noindent
{\bf Remark:} We can adapt this algorithmic framework to the egalitarian objective or Pareto efficiency, instead of welfare maximization. 

For the egalitarian objective, for each action profile, instead of regret vector $r(a)$, we can consider $(nm(m-1) + n)$-dimensional vector $\rho(a)$. The first $nm(m-1)$ components of $\rho(a)$ and $r(a)$ are the same. But, the last $n$ components of $\rho(a)$ are set equal to $w' - u_p(a)$ for each $p \in [n]$. Here $w'$ is the optimal value of the egalitarian objective, $w' := \argmax_{x \in \textrm{CE }} \min_p u_p(x)$. Working with $\rho(a)$ and the corresponding modified-welfare function, we can obtain an $\varepsilon$-CE $x$ that satisfies $\min_p u_p(x) \geq w' - \varepsilon$. 

To find an approximate correlated equilibrium that is nearly Pareto efficient, we pick a specific player $q$ and replace the last component of the regret vector $r(a)$  by $w'' - u_q(a)$. Here $w'' : = \argmax_{x \in \textrm{CE }}  u_q(x)$. In this case, we can consider the relevant  modified-welfare function and overall obtain an $\varepsilon$-CE that satisfies $u_q(x) \geq w'' - \varepsilon$. Since there does not exist a CE wherein the utility of $q$ is $\varepsilon$ more  than $u_q(x)$, we get that $x$ is $\varepsilon$-Pareto efficient.

\subsection{Aggregative Games}
\label{sect:agg-game}
This section presents a polynomial-time additive-approximation algorithm for MWMP in aggregative games. An $n$-player $m$-action aggregative game with action profiles $A$ is specified by an aggregator function $S: A \rightarrow [-W, W]^k$ and utility-defining functions $v_p: A_p \times [-W,W]^k \rightarrow [0,1]$. The function $S$ serves as a sufficient statistic for the utilities of the player; specifically, the utility of player $p$ at action profile $a$ (i.e., $u_p(a)$) is equal to $v_p(a_p, S(a))$. Note that here the utility depends on the action of the player, $a_p$, and the aggregated vector, $S(a)$. 
In aggregative games, the function $S$ is additively separable; in particular, there exist vectors  $f_p(a_p) \in [-W', W']^k$ for each player $p \in [n]$ and action $a_p \in A_p$ such that for any action profile $a \in A$ we have the following component-wise equality: $S(a) = \sum_{p} f_p(a_p)$.  Here, the dimension $k$ is assumed to be a fixed constant and $W$ and $W'$ are polynomially bounded in $n$ and $m$.

Along the lines of prior work (see~\cite{babichenko2013best, cummings2014privacy}), we consider the setting in which the influence of the aggregator on the utilities is bounded: $|v_p(a_p, s) - v_p(a_p, s') | \leq \| s - s' \|_{\infty} $ for all players $p \in [n]$, actions $a_p \in A_p$, and vectors $s, s' \in [-W, W]^k$. The assumption that $k$ is a fixed constant can be mute without this bounded influence property.  

Recall that modified utilities $\tilde{u}_p^y(a)$ are defined in terms of the utilities $u_p(a)$; see Definition~\ref{def:mod-wel}. In this section we will only consider vectors $y$ whose components are linearly bounded, i.e., $y \in [0,n]^d$;  in order to apply Theorem~\ref{thm:graddes} it suffices to consider such linearly-bounded vectors.  

We begin by discretizing the aggregating vectors $f_p(a_p)$ such that their components are multiples of parameter $\delta$, which will be set appropriately. That is, for all $p \in [n]$ and $a_p \in A_p$, the components of vector $f_p(a_p)$ are rounded to the nearest multiple of $\delta$. Note that component-wise the vectors $f_p(a_p)$ are polynomially bounded; hence, a polynomially small $\delta$ ensures that even after discretization for all action profiles $a$, the aggregated value $S(a)$ remains within $ O \left( \frac{ \varepsilon^4 }{poly(n,m)} \right)$---under the $\ell_\infty$ norm---of the original (undiscretized) value. The bounded influence assumption, $|v_p(a_p, s) - v_p(a_p, s') | \leq \| s - s' \|_{\infty} $ and the fact that $y \in [0,n]^d$ ensures that the discretization process does not change the modified utility $\tilde{w}^y(a)$ by more than $ O \left( \frac{ \varepsilon^4 }{poly(n,m)} \right)$, for any action profile $a$. Overall, this implies that (with a polynomially small $\delta$) if we compute an action profile $a'$ that maximizes $\tilde{w}^y$ with the discretized aggregator function then $a'$ is an $ O\left( \frac{ \varepsilon^4 }{n^4m^8} \right)$-additive approximation for MWMP with the original (undiscretized) aggregator function. 

Throughout the remainder of the section we will work with the discretized aggregator. Now, all the discretized vectors $f_p(v_p)$s are contained in $\left\{0, \pm \delta, \pm 2 \delta, \pm 3 \delta, \ldots, \pm \left\lceil \frac{W'}{\delta} \right\rceil \delta \right\}^k$. Write $\mathcal{G}$ to denote the $k$-dimensional grid defined as follows: $\mathcal{G} := \{ \sum_{q= 1}^p f_q(a_q) \mid \textrm{ for all } p \in [n] \textrm{ and each } a_q \in A_q \}$. Since $\delta$ is polynomially small we have $|\mathcal{G}| = \left(\frac{nm}{\varepsilon}\right)^{O(k)}$. Also, for all action profiles $a \in A$, the discretized aggregator function value $S(a)$ is contained in $\mathcal{G}$. 

We develop a dynamic program that works over $\mathcal{G}$ and computes an action profile that maximizes the modified welfare $\tilde{w}^y$. As discussed above, this gives us an $ O\left( \frac{ \varepsilon^4 }{n^4m^8} \right)$-additive approximation for MWMP. 

Our main result of this section is that an additive approximation for MWMP can be computed efficiently when the scaling vector $y$ is contained in  $[0,n]^d$. 

\begin{theorem}
\label{thm:agggame}
Given an $n$-player $m$-action aggregative game and a scaling vector $y \in [0,n]^{(nm(m-1) +1)}$, there exists a polynomial-time algorithm that computes an $ O\left( \frac{ \varepsilon^4 }{n^4m^8} \right)$-additive approximation for the MWMP instance specified via $y$.
\end{theorem}

\begin{proof}
Throughout the proof we work with the modified welfare function $\tilde{w}^y$ that is specified by the discretized aggregator function. In particular, via a dynamic program we will compute an action profile that maximizes $\tilde{w}^y$. As mentioned above, this gives  the desired additive-approximation guarantee. 

For each vector $s \in \mathcal{G}$ we will maximize the modified welfare $\tilde{w}^y$ function over the set of action profiles $A(s) := \{ a \in A \mid S(a) = s\}$ in polynomial time. Write $a^*$ to denote an action profile that maximizes $\tilde{w}^y$. Discretization ensures that $S(a^*) \in \mathcal{G}$. Also, the fact that the cardinality of $\mathcal{G}$ is polynomially bounded implies that efficiently optimizing over $A(s)$ for each $s \in \mathcal{G}$ gives an action profile that maximizes $\tilde{w}^y$ in polynomial time.

The remainder of the proof details an algorithm that, given vector $s \in \mathcal{G}$, solves the following optimization problem in polynomial time: $\argmax_{a \in A(s)} \tilde{w}^y(a) $.

First, we define modified utility function $\tilde{v}_p(a_p, s)$ in terms of the given functions $v_p(a_p,s)$. Recall that the vector $y$ is $d=nm(m-1) + 1$ dimensional and its first $nm(m-1)$ components are indexed by triples $(p,i,j)$ with $p \in [n]$ and distinct $i,j \in [m]$. 
Using vector $y$ as a parameter we define, 
\begin{align*}
\tilde{v}_p^y(a_p, s) & := y_d \ v_p(a_p,s) +  \sum_{j \in A_p} y_{(p, a_p, j) } \cdot \left[  v_p(a_p,s)  - v_p(j, s - f_p (a_p) + f_p(j) \right]
\end{align*}

A key observation is that for any action profile $a$, if $s = S(a)$ then $\tilde{u}^y_p(a) = \tilde{v}^y_p(a_p,s)$ and, hence, $\tilde{w}^y(a) = \sum_{p=1}^n \tilde{v}^y (a_p, s)$. Therefore, by the definition of $A(s)$, for all $a \in A(s)$ the following equality holds $\tilde{w}^y(a) = \sum_{p=1}^n \tilde{v}^y (a_p, s)$. Hence,  $\argmax_{a \in A(s)} \sum_{p=1}^n \tilde{v}^y (a_p, s)$ is equal to $\argmax_{a \in A(s)} \tilde{w}^y(a)$.

We solve $\argmax_{a \in A(s)} \sum_{p=1}^n \tilde{v}^y (a_p, s)$ via a dynamic program that fills a matrix $M(p, s')$ indexed by $p \in [n]$ and $s' \in \mathcal{G}$. In particular, $M(p, s')$ is set equal to $ \max_{a_1, a_2, \ldots, a_p } \{  \sum_{q=1}^p \tilde{v}^y (a_q, s) \mid \sum_{q=1}^p f_q(a_q) = s' \}$. Here, the entry $M(n, s)$ is equal to the target optimal value $\max_{a \in A(s)} \tilde{w}^y(a)$.
We can initialize $M(1, s')$ by going over actions in $A_1$; specifically, $M(1, s') = \max_{a_1 }\{ \tilde{v}^y(a_1, s) \mid  f_1(a_1) = s'\}$. In general, we use the recurrence relation $M(p, s') = \max_{a_p, s''} \{  \tilde{v}^y(a_p, s)  + M(p-1, s'') \mid f_p(a_p) + s'' = s' \}$ to complete the matrix. A direct inductive argument proves the correctness of this dynamic program. 

Since, $|\mathcal{G}|$ is polynomially bounded, the size of the matrix is also polynomially bounded. Overall, the dynamic program runs in polynomial time and the stated claim follows.   
\end{proof}

Theorem~\ref{thm:agggame} shows that the additive approximation required in Theorem~\ref{thm:graddes} can be computed in polynomial time. Therefore, the two theorems together imply that in aggregative games an approximate correlated equilibrium with near-optimal  welfare can be computed in polynomial time. 

\begin{corollary}
In $n$-player $m$-action aggregative games an $\varepsilon$-correlated equilibrium with welfare $w^* - \varepsilon$ can be computed in time polynomial in $n$, $m$, and $1/\varepsilon$. Here $w^*$ denotes the optimal welfare over the set of correlated equilibria of the given game.
\end{corollary}
